\documentclass[reprint, prl, aps, amssymb, notitlepage, twocolumn]{revtex4-1}

\usepackage{mathtools}
\usepackage{braket}

\usepackage{graphicx}

\usepackage[utf8]{inputenc}
\usepackage[main=english, czech]{babel}
\usepackage{bm}
\usepackage{natbib}
\usepackage{amsthm}

\usepackage{proof}

\usepackage{url}
\usepackage{tikz}
\usetikzlibrary{positioning}
\usepackage{proof}

\newtheorem{claim}{Claim}
\newtheorem{conjecture}{Conjecture}
\usepackage{complexity}

\definecolor{ltgray}{rgb}{0.95,0.95,0.95}
\usepackage{textcomp}

 \usepackage{array}

\usepackage{hyperref}
\hypersetup{
    colorlinks=true,
    linkcolor=blue,  
    citecolor=blue,
    urlcolor=blue,
}

\begin{document}
\title{Simple Communication Complexity Separation \\ from Quantum State Antidistinguishability}
\author{\foreignlanguage{czech}{Vojtěch Havlíček}}
\email{vojtech.havlicek@keble.ox.ac.uk}
\author{Jonathan Barrett}
\affiliation{Department of Computer Science, University of Oxford, Wolfson Building, Parks Road, Oxford
OX1 3QD, UK} 

\begin{abstract}
A set of $n$ pure quantum states is called antidististinguishable if there exists an $n$-outcome measurement that never outputs the outcome `$k$' on the $k$-th quantum state. We describe sets of quantum states for which any subset of three states is antidistinguishable and use this to produce a two-player communication task that can be solved with $\log d$ qubits, but requires one-way communication of at least $\log (4/3) (d-1) - 1 \approx 0.415 (d-1) - 1$ classical bits. The advantages of the approach are that the proof is simple and self-contained -- not needing, for example, to rely on hard-to-establish prior results in combinatorics -- and that with slight modifications, non-trivial bounds can be established in any dimension $\geq 3$. The task can be framed in terms of the separated parties solving a relation, and the separation is also robust to multiplicative error in the output probabilities. We show, however, that for this particular task, the separation disappears if two-way classical communication is allowed. Finally, we state a conjecture regarding antidistinguishability of sets of states, and provide some supporting numerical evidence. If the conjecture holds, then there is a two-player communication task that can be solved with $\log d$ qubits, but requires one-way communication of $\Omega (d \log d)$ classical bits.
\end{abstract}

\maketitle

\paragraph{Introduction.}
How difficult is it to communicate classically the identity of a quantum state in an entanglement-unassisted scenario? Specifying a pure state of a qubit requires two real numbers, so communicating its identity seemingly needs an infinite amount of classical information \cite{Cerf00}. This suspicion is confirmed by the results of  \cite{Hardy04}, which show that if the two communicating parties do not share random data, an unbounded amount of classical communication is needed to exactly simulate results of quantum experiments.  Assuming shared randomness however, only two bits suffice to exactly reproduce results of any projective qubit measurement  \cite{Toner03}. 

Here we derive a lower bound for the classical communication cost of simulating the transmission of a $d$-dimensional quantum state. This is done by describing a communication problem based on a relation, and proving an exponential separation between the classical and quantum communication complexities. The proof uses quantum state antidistinguishability, a concept that has been studied in the foundations literature \cite{PBR12, Caves02, Kerppo18, Leifer14, Barrett14}, and lies behind the theorem of Ref.~\cite{PBR12}, which rules out $\Psi$-epistemic ontological models of quantum theory. 

It is already well known that there can be exponential separation between classical and quantum communication complexities. Ref.~\cite{Buhrman98} (see also Ref.~\cite{Brassard99}), for example, presents an $\Omega({d})$ vs $O(\log d) $ separation between classical and quantum zero-error communication complexities for the task of deciding whether two $d$-bit inputs are either equal or have Hamming distance $d/2$. The proof presented involves a highly non-trivial combinatorial result \cite{Frank87}, and results in a classical lower bound of the form $c d$, for a constant $c = 0.01$. (Slightly strengthened, the same proof yields a lower bound of $c d$ for any $c < 0.02$, which means that a non-trivial separation between the number of classical bits that must be communicated and the number of qubits can be established for any $d\geq 512$.) Subsequent works \cite{Raz99, Ambainis03, BarYossef08, Gavinsky06b, Kerenidis06, Gavinsky06, Gavinsky08, Klartag10, Montanaro11} established exponential separations between zero-error quantum and bounded-error classical protocols. The separations of Refs.~\cite{Gavinsky08, Klartag10} in particular, are strong in the sense that the separation holds between zero-error one way quantum protocols and bounded-error classical protocols that allow two-way communication. In none of Refs.~\cite{Raz99, Ambainis03, BarYossef08, Gavinsky06b, Kerenidis06, Gavinsky06, Gavinsky08, Klartag10, Montanaro11}, however, is a classical lower bound of $\Omega(d)$ established: varying separations are presented of which the strongest is $\Omega(\sqrt{d})$ vs $O(\log d) $. Ref.~\cite{Montanaro16q}, by considering a task based on distributed Fourier sampling, derives an $\Omega(d)$ vs $O(\log d)$ separation, which is robust against constant additive error, and which holds when two-way classical communication is allowed. For a 2010 review, see Ref.~\cite{Buhrman10}.

Montina \cite{Montina11}, considers the scenario in which Alice's input can be any pure state of a $d$ dimensional quantum system, and Bob's input can be any two-outcome measurement consisting of a rank $1$ projector and the orthogonal projector. Assuming zero-error, one-way classical communication, Montina derives a classical lower bound of $c d$, for $c \approx 0.293$, where the proof uses a result concerning volumes of subsets of a hypersphere due to Raigorodskii \cite{Raigorodskii99}. 
It is also shown that a classical lower bound of $d-1$ follows from (a complex generalization of) a mathematical conjecture known as the \emph{double cap conjecture}.

Finally, other works have had slightly different aims from that of establishing quantum-classical separations in the standard communication complexity setting, but nonetheless contain techniques related to those that we use. These include Ref.~\cite{Buhrman01}, which presents quantum fingerprinting protocols, and Refs.~\cite{Perry15, ZiWen16}, where the notion of antidistinguishability is used to give separations between one-way communication and information complexities in exclusion games. They also include Ref.~\cite{Karanjai18}, where lower bounds on the size of a classical memory needed to simulate quantum processes are derived, and applied to the stabilizer subtheory of quantum theory. Ref.~\cite{Heinosaari19} defines and studies tasks of communicating `partial ignorance', including communication tasks using antidistinguishable quantum states that are similar to those used here.  

One of the main motivations for our work is to present a proof of exponential separation between classical and quantum communication complexity that is very simple, and self-contained. Aside from this, the advantages of the approach include (i) a lower bound for zero-error one way classical communication of $c d$, with $c\approx 0.415$, which is the strongest that we have seen, and (ii) a separation between quantum and classical one-way communication complexity for any $d\geq 3$. On the negative side, we show that the classical communication lower bound disappears if two-way classical communication is allowed. Although the result is robust against a limited amount of additive noise, the lower bound also disappears if bounded error classical protocols are allowed.

We use asymptotic $O$-notation throughout \cite{Cormen09}. All logarithms are base $2$, $[n]$ denotes the set $\lbrace 1,2, \ldots , n\rbrace$ and $\lbrace 0, 1 \rbrace^*$ is the set of all finite bit-strings.

\paragraph{Communication complexity.}

Communication complexity studies the amount of communication needed to solve distributed computational problems \cite{Yao79, Yao93, Kushilevitz96, Rao19}. In a two-party relational task, Alice and Bob get inputs $x \in X$ and $y \in Y$ respectively, for finite sets $X,Y$, and do not see the other's input. The aim is for Bob to output $z \in Z$, such that $(x,y,z) \in R$ for a relation $R \subseteq X \times Y \times Z$. Both parties can use unlimited computational power and exchange messages following a shared communication protocol. In this work, we allow shared randomness, meaning that Alice and Bob share a random string $s \in \lbrace 0, 1 \rbrace^*$ sampled from a distribution $P(s)$. On any run of the protocol, the classical (quantum) communication cost is the number of transmitted bits (qubits). In general, this can depend on both the input and the value of $s$. The notion of communication complexity we use is the communication cost in the best possible protocol, where the communication cost is averaged over the shared random data, and evaluated on the worst-case input. Note that with this definition, the communication complexity with shared random data can be smaller than the deterministic communication complexity (where no randomness is permitted), even for zero error protocols \cite{Kushilevitz96}. The communication cost for the worst-case value of $s$ is always larger than that averaged over $s$, hence our lower bounds for classical communication complexity still apply if communication complexity is defined with respect to the worst-case value of $s$. One-way communication complexity assumes a protocol in which Alice is only allowed to send a single message to Bob, after which he announces the result. 
 
\paragraph{Antidistinguishability}

The quantum protocol that we will describe relies on \emph{antidistinguishable} sets of quantum states. A set of $n$ pure quantum states $\ket{\rho_1}, \ket{\rho_2}, \ldots , \ket{\rho_n}$ is antidistinguishable if there exists an $n$-outcome measurement $\Pi'  = \lbrace \Pi_{z}'  \, |  \; z \in [n] \rbrace$ that never outputs the outcome $z$ on the quantum state $\ket{\rho_z}$, i.e., 
\begin{align}
\Pi_{z}' \ket{\rho_z} &= 0, \; \forall\, z \in [n]. 
\label{Eq:Antidistinguishability} 
\end{align}
A sufficient condition for three pure quantum states $\ket{\rho_j}, \ket{\rho_k}, \ket{\rho_m}$ to be antidistinguishable with a projective measurement is if there exist orthogonal states $\ket{j}, \ket{k}, \ket{m}$, such that: 
\begin{equation}
\begin{aligned}
\ket{\rho_j} &= \cos \theta_i \ket{k}  +  e^{i\phi_i} \sin \theta_i \ket{m}, \\
\ket{\rho_k} &= \cos \theta_j \ket{m}  + e^{i\phi_j} \sin \theta_j \ket{j}, \\
\ket{\rho_m} &=  \cos \theta_k \ket{j}  + e^{i\phi_k} \sin \theta_k \ket{k},
\end{aligned}
\label{Eq:PPIncomp}
\end{equation}
for some $\theta_z, \, \phi_z \in \mathbb{R}$, $z \in \lbrace j,k,m \rbrace$. The notion of antidistinguishability was introduced by Caves, Fuchs and Schack in Ref.~\cite{Caves02}, where it is shown that  that such a basis can be found iff for:
 \begin{align*} a &= |\braket{\rho_j|\rho_k}|^2, & b &= |\braket{\rho_j|\rho_m}|^2, & c &= |\braket{\rho_m|\rho_k}|^2,\end{align*}
the following holds \footnote{As also noted in Ref.~\cite{Stacey14}, Ref.~\cite{Caves02} contains a minor error, in which $>$ instead of $\geq$ appears in the second part of the condition.}:
\begin{align*}
a+b+c &<  1, &
(1-a-b-c)^2 &\geq 4abc.
\end{align*}
As a simple corollary, any triple of pure quantum states is antidistinguishable if:
\begin{align}  a,b,c \leq \frac{1}{4}.
\label{Eq:AD}
\end{align}
Now consider a finite set $S$ of pure states $\ket{\rho_1}, \ket{\rho_2}, \ldots , \ket{\rho_{|S|}}$, for which: \begin{align*} |\braket{\rho_i|\rho_j}| \leq \delta; \; \forall \, i \neq j, \end{align*} where $\delta \in [0,1)$. Such sets are called \emph{complex spherical codes}, and have been much studied, with applications in classical signal processing, error correction, and quantum information \cite{Welch74, Levenshtein78, Levenshtein83, Buhrman01, RenesPHD04, Roy11, Zorlein15, Conde17}. Our results will be obtained from the following simple observation:
\begin{claim} 
Any triple of states drawn from a  complex spherical code $S$ with $\delta \leq 1/2$ is antidistinguishable. \label{Corollary:Triples} 
\end{claim}

\paragraph{Separations.}

The separation between classical and quantum communication complexities that we establish is for solution of a relational task, as follows. For a $d$ dimensional Hilbert space, let $S$ be a complex spherical code $S = \{ \ket{\rho_1}, \ket{\rho_2}, \ldots \ket{\rho_{|S|}} \} $, with $\delta = 1/2$. Alice's input is an integer $i \in [|S|]$. Bob's input is a set of $3$ integers $j,k,m \in [|S|]$. Bob must output one of the integers $j,k,m$, under the constraint that his output must not be equal to Alice's input. Setting $X = Z = [|S|]$, and $Y = \{ \, T\,  | \, T\subseteq S, |T|=3\, \}$, the relation is thus given by
\begin{equation}
\begin{aligned}
R \subseteq X \times Y & \times Z: \\ 
(i , \{ j,k,m\} , z) & \in R \text{ iff } z \in \{ j,k,m \}  \mathrm{\ and\ } z \ne i. \label{Eq:Relation}
\end{aligned}
\end{equation}  

The quantum solution is straightforward (see Fig.~\ref{Fig:Task}). 
\begin{enumerate}
\item Given input $i$, Alice prepares a quantum system in the state $\ket{\rho_i} \in S$ and sends it to Bob. 
\item Given input $\{ j,k,m\}$, Bob performs an antidistinguishing measurement for the three states  $\ket{\rho_j} ,  \ket{\rho_k},  \ket{\rho_m}$ on the system he receives from Alice. Label the three outcomes $\Pi'_{z}$, for $z \in \{j,k,m \}$, such that $\Pi'_{z} \ket{\rho_z} = 0$. If Bob obtains outcome $\Pi'_{z}$, then he outputs $z$. ({This means that Bob never outputs an outcome `$i$'.})
\end{enumerate}
The communication complexity is the number of qubits transmitted, which is $\lceil \log d \rceil$. 
\begin{figure}[t]
\begin{tikzpicture}
\node[circle, draw, thick, inner sep=2pt] (A) at(-1.5,0) {\small \text{A}};
\node[circle, draw, thick, inner sep=2pt] (B) at(1.5,0) {\small \text{B}};
\node[rectangle] (S) at (-2,-1) {$i$};
\node[rectangle] (M) at (2,-1) {$\{j,k,m\}$};
\node[rectangle] (m) at (2.5,0) {$z$};
\draw[->] (A) --node [above,midway] {$\ket{\rho_i}$}  (B);
\draw[->](S) -- (A);
\draw[->](M) -- (B);
\draw[dashed, ->](B) -- (m);
\end{tikzpicture}
\caption{A quantum protocol. Alice gets an integer $i \in [|S|]$, and Bob gets three integers $j,k,m \in [|S|]$. Alice sends a quantum system in the state $\ket{\rho_i}$ to Bob. Bob performs an antidistinguishing measurement for the states $\ket{\rho_j} , \ket{\rho_k} , \ket{\rho_m}$, and outputs the outcome.}
\label{Fig:Task}
\end{figure}
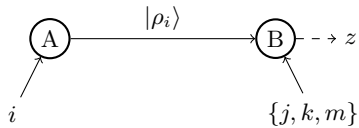
A classical protocol is illustrated in Fig.~\ref{Fig:Protocol}.
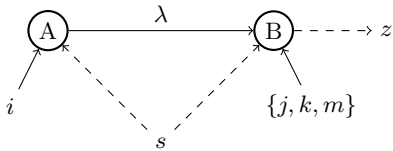
\begin{figure}[t]
\begin{tikzpicture}
\node[circle, draw, thick, inner sep=2pt] (A) at(-1.5,0) {\small \text{A}};
\node[circle, draw, thick, inner sep=2pt] (B) at(1.5,0) {\small \text{B}};
\node[rectangle] (S) at (-2,-1) {$i$};
\node[rectangle] (M) at (2,-1) {$\lbrace j,k,m \rbrace$};
\node[rectangle] (s) at (0, -1.5) {$s$};
\node[rectangle] (m) at (3,0) {$z$};
\draw[->] (A) --node [above,midway] {$\lambda$}  (B);
\draw[->](S) -- (A);
\draw[->, dashed](s) -- (A);
\draw[->](M) -- (B);
\draw[->, dashed](s) -- (B);
\draw[dashed, ->](B) -- (m);
\end{tikzpicture}
\caption{A classical protocol. Alice gets an integer $i \in [|S|]$, and Bob gets three integers $j,k,m \in [|S|]$. Alice sends a message $\lambda$, after which Bob outputs $z \in \lbrace j,k,m \rbrace$. The classical strategy can use a shared random string $s$, drawn according to a probability distribution $P(s)$. }
\label{Fig:Protocol}
\end{figure}

\begin{claim}
\label{Thm:OneWay}
The zero-error one-way classical communication complexity of the task is at least $\lceil \log |S| -1 \rceil$. 
\end{claim} 
\begin{proof}
In order to establish the claim, consider first deterministic protocols, in which the message that Alice sends to Bob is a function of her input, and Bob's output is a function of his input and the message. Suppose that there are three distinct values of Alice's input, $ \ket{\rho_j}, \ket{\rho_k}, \ket{\rho_m}$, such that the same message $\lambda$ is sent for each of them. If Bob receives $\lambda$, and his input is the triple $\{ j,k,m \}$, then there is no output he can give that will not sometimes generate an error. Therefore Alice needs at least a distinct message per two states of $S$. This gives:
\begin{align}
|\Lambda| \geq \frac{|S|}{2},\label{Eq:LowerBound}
\end{align}
where $\Lambda$ is the set of possible values of Alice's message. It follows that on the worst case input, Alice needs to send at least $\lceil \log |S|- 1 \rceil$ bits. 

In the presence of shared randomness $s$, each value of $s$ defines a deterministic protocol, which in the zero-error case must respect the relation. If communication complexity is evaluated as the communication cost on the worst case values of $s$, this concludes the proof. With communication complexity given by the communication cost averaged over $s$, and evaluated on the worst case input, the following standard manoeuvre \cite{Kushilevitz96} suffices. For any probability distribution $Q$ over input pairs, the communication complexity is lower bounded by the communication cost, averaged both over values of $s$ and over inputs drawn according to the distribution $Q$. This is in turn lower bounded by the communication cost, averaged over inputs drawn according to the distribution $Q$, of the deterministic protocol that achieves the lowest value for this cost. Choosing $Q$ as the uniform distribution over all input pairs, the argument above establishes that the lower bound of $\lceil \log |S|- 1 \rceil$ bits still holds. Same bound straightforwardly applies to a multiplicative error sampling task outlined in Appendix A.
\end{proof}

The lower bound of Claim~\ref{Thm:OneWay} is determined by $|S|$, the size of the spherical code $S$.  In contrast, for a fixed dimension $d$, regardless of the size of $|S|$, the quantum protocol uses only $\lceil \log d \rceil$ qubits. This gives a communication advantage whenever $|S| > 2d$.  For the best separation, the problem becomes: how large can a complex spherical code in $d$ dimensions be, with $\delta = 1/2$?

In $d=3$, the largest such set is given by an equiangular complex spherical code, otherwise known as a symmetric, informationally-complete set (SIC) \cite{Renes04, Scott09, Scott17, Stacey14}.
A SIC in dimension $3$ consists of $9$ unit vectors such that: 
\begin{align}
 |\braket{\rho_i | \rho_j}|^2 = \frac{1}{4}, \; \forall \, i \neq j.
 \label{Eq:SIC}
\end{align}
That a larger set cannot be found follows from the Welch bound \cite{Welch74}, which states that: 
\begin{align}
\max_{i \neq j} |\braket{\rho_i | \rho_j}|^2 \geq \frac{|S| - d}{d(|S|-1)}, \label{Eq:Welch}
\end{align}
for any set of $d$-dimensional pure states $S = \lbrace \ket{\rho_1}, \ket{\rho_2} ,\ldots \ket{\rho_{|S|}} \rbrace$. This shows that Alice needs to send at least a $5$-valued message vs a $3$-dimensional quantum system, or in terms of whole numbers of bits and qubits, at least $3$ bits vs. $2$ qubits. 

In $d\geq 4$, mutually unbiased bases (MUBs) yield larger sets than SICs. It is known that for $d$ power prime, there exist $d+1$ distinct MUBs  \cite{Wooters89}, which satisfy $|\braket{\rho_i | \rho_j}|^2 \leq 1/d$. Taking $S$ as the union of the vectors in the MUBs gives $|S| = d(d+1)$, hence a $\lceil \log(d^2 + d) - 1 \rceil$ lower bound on classical communication. In $d=4$, Alice needs to send at least a $10$-valued message vs. a $4$-dimensional quantum system, or in terms of whole numbers of bits and qubits, at least $4$ bits vs. $2$ qubits. A bound due to Levenshtein \cite{Levenshtein78, Levenshtein83, Zorlein15} implies that in $d=4$, the $20$-element set of vectors given by MUBs is the largest that can be achieved with $\delta =1/2$. In $d\geq 5$, larger sets than those given by MUBs have been found numerically \cite{Zorlein15, Conde17}.

For general $d$, the following supplies a classical lower bound of $\lceil \log (4/3) (d-1) -1 \rceil \approx 0.415 (d - 1) - 1$ bits.
\begin{claim}
For any $d$, there exists a complex spherical code, with $\delta=1/2$, such that $|S| \geq \left( \frac{4}{3} \right)^{d-1}$.
\label{Claim:LargeCodes}
\end{claim}
\begin{proof}
The claim is established by generalizing a well known result of Chabauty, Shannon and Wyner \cite{Chabauty53, Shannon59, Wyner65, Jenssen18} to the case of complex vector spaces. Consider, for each vector $|e\rangle \in S$, the complex spherical cap $A^d_{\theta}$, defined as the set of all vectors $\ket{\psi}$ in the Hilbert space satisfying $| \langle e | \psi \rangle |^2 \geq \cos^2 \theta$, for some $0 \leq \theta \leq \pi/2$. If $S$ is as large as possible under the constraint $\delta = 1/2$, then for $\theta = \pi/3$, these caps must cover the whole of the complex unit sphere -- otherwise there is room to add another vector to $S$. Therefore, a simple lower bound on the achievable $|S|$ is given by
\begin{equation}\label{volumesforlowerbound}
|S| \geq \frac{V^d}{V^d_{\pi/3}},
\end{equation}
where $V^d_{\theta}$ is the volume of a spherical cap $A^d_{\theta}$, and $V^d = V^d_{\pi/2}$ is the volume of the unit sphere in $d$ complex dimensions, volumes being evaluated according to some suitable measure.

The following calculation (with different $\theta$, in the context of a different method for establishing a communication complexity separation) appears in Ref.~\cite{Montina11}. In keeping with our main motivation of providing a simple and self-contained proof of exponential separation, we reproduce the reasoning here.

The points of the unit sphere in $d$ complex dimensions are in $1-1$ correspondence with the points of the unit sphere in $2d$ real dimensions, under the obvious mapping that takes the real and imaginary parts of each complex coordinate to two independent classical coordinates. Volumes of subsets of the complex unit sphere can therefore be defined as the volumes of the corresponding subsets of the real unit sphere in $2d$ dimensions. Letting $\ket{e} = (1,0,\ldots , 0) \in \mathbb{C}^d$, the complex spherical cap $A^d_{\theta}$ maps to the set of real vectors of the form:
\begin{equation}
\cos \phi \, \hat{u}_1 + \sin \phi \, \hat{u}_2,
\end{equation}
where $\hat{u}_1 \in \mathbb{R}^{2d}$ is a unit vector in the subspace spanned by $(1,0,0, \ldots ,0)$ and $(0,1,0,\ldots , 0)$, $\hat{u_2}$ is a unit vector in the orthogonal subspace, and $0\leq \phi \leq \theta$. 
The volume of this set is given by:
\begin{equation}\label{expressionforvolume}
V^d_\theta =   \int_0^{\theta} 2\pi \cos \phi \ \tilde{V}^{2d-2}(\sin\phi ) \ \mathrm{d}\phi,
\end{equation}
where $\tilde{V}^{2d-2}(\sin\phi )$ is the volume of a $(2d-3)$-sphere in $(2d-2)$ real dimensions of radius $\sin\phi$. 
Combining Eqs.~(\ref{volumesforlowerbound}) and (\ref{expressionforvolume}), with $\theta = \pi/3$, yields
\[
|S| \geq \left( \frac{4}{3} \right)^{d-1} .
\]
\end{proof}

An alternative proof that there exist sets $S$ with $|S|$ exponentially large, which results in a worse bound, but which some may find even simpler, is given in Ref.~\cite{Buhrman01}, and reproduced in Appendix B.

\paragraph{Two-way classical communication.}
\begin{claim}
Assuming $|S| = 2^q$, for integer $q$, the two-way classical communication complexity of $R$ is  at most $\lceil \log \log |S| \rceil + 1$ bits. 
\label{Claim:TwoWay}
\end{claim}
\begin{proof}
The assumption that $|S| = 2^q$ is not essential, but allows a short statement of the proof. The protocol has two rounds and starts with a message from Bob to Alice. Let Bob's input be $\lbrace j,k,m \rbrace $, and assume without loss of generality that $j < k < m$. Bob determines the largest integer $r \leq q = \log |S| $, such that for some integer $s \geq 0$ either:
\begin{equation}
\begin{aligned}
s \, 2^{r} &< j,k \leq (s+1)\, 2^{r} \text{ and } \\
(s+1)\, 2^{r} &< m \leq (s+2)\, 2^{r},
\end{aligned}
\label{Eq:Case1}
\end{equation}
or 
\begin{equation}
\begin{aligned}
s \, 2^{r} &< j \leq (s+1)\, 2^{r} \text{ and } \\
(s+1)2^{r} &< k,m \leq (s+2)2^{r}.
\end{aligned}
\label{Eq:Case2}
\end{equation}
Bob sends $r$ to Alice using $\lceil \log\log|S| \rceil $ bits. Note that $r$ determines a subset $Y_r \subseteq Y$ of Bob's possible inputs. For input $i\in  [|S|]$, Alice computes the parity $p$ of $\left\lceil \frac{i}{2^r} \right\rceil$, and sends it to Bob. Note that the set $X$ of Alice's inputs is partitioned by $p$ and $r$ into subsets $X_{r,p} \subseteq X$. By Eqs.~(\ref{Eq:Case1}) and (\ref{Eq:Case2}), at least one of $j,k,m$ is not in $X_{r,p}$. Bob chooses such a value for his output. Communicating $\lceil \log \log |S| \rceil + 1$ bits hence suffices in the two-way scenario. Fig.~\ref{Fig:Upsilon} illustrates the protocol for $|S|=8$. \end{proof}
\begin{figure*}[t]
\includegraphics[scale=.39]{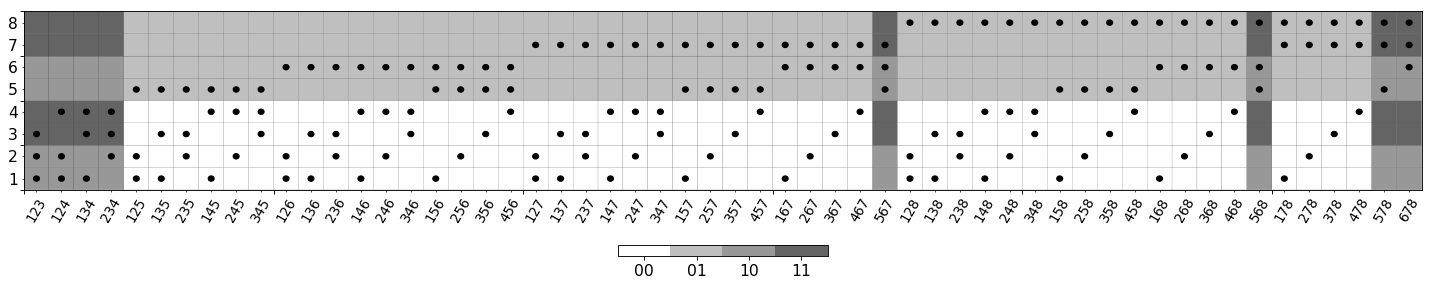}
\caption{An illustration of the two-way classical protocol with $|S|=8$. Rows correspond to different values of Alice's input, columns to different values of Bob's input. In single run of the protocol, Bob sends an integer $r$ to Alice, and Alice sends an integer $p$ to Bob. Each value of the conversation $(r,p)$ is compatible with a subset of the joint inputs, where the subset is of the form $X_{r,p} \times Y_r$ for $X_{r,p}\subset X$ and $Y_r \subset Y$, and is known as a combinatorial rectangle. In the $|S|=8$ example, there are four possible conversations, hence four rectangles, which cover $X\times Y$ as illustrated. No rectangle contains three dots in the same column, which implies that Bob can always produce a suitable output without error.} 
\label{Fig:Upsilon}
\end{figure*}

\paragraph{Classical bounded error.}

The separation also disappears with one-way classical communication, if the classical players need only solve the relation with bounded error. Roughly stated, the argument is that Alice can simply partition $S$ into $K$ subsets, and send $\lceil \log K \rceil$ bits to Bob to indicate in which subset her input lies. In a deterministic protocol, this risks an error whenever Bob's input $\{j,k,m\}$ corresponds to three states all within the same subset. However, the probability of error on the worst case input pair can be reduced to $1/K^2$ by allowing the shared random data to specify a partition, and varying over all partitions into $K$ subsets.

\paragraph{Conclusion.} 

We have used quantum state antidistinguishability to give simple proofs of separation between classical and quantum zero-error one-way communication complexities. Using SICs and MUBs, we proved separations in any $d\geq 3$. Using a lower bound on the achievable size of a suitable complex spherical code, we showed an exponential separation of $\log d$ qubits vs. $0.415 (d-1) - 1$ bits. For the relation considered, however, the separation disappears if two-way classical communication is allowed or if one-way classical communication is allowed with bounded error. 

The results are stated in terms of quantum and classical players solving a relation, which is defined on finite sets of possible inputs for Alice and Bob. Seeing as expanding the sets of inputs can only make things more difficult for classical players, the lower bounds also apply to the one-way classical communication cost of simulating an experiment in which Alice prepares an arbitrary pure state of a $d$-dimensional quantum system, and sends it to Bob, who performs an arbitrary projective measurement. (Naturally, the same can be said for the results of any of Refs.~\cite{Buhrman98, Brassard99, Raz99, Ambainis03, BarYossef08, Gavinsky06b, Kerenidis06, Gavinsky06, Gavinsky08, Klartag10, Montanaro11, Montanaro16q, Montina11}.) In this more general scenario, if error $\epsilon$ is tolerated in the measurement outcome probabilities in the classical simulation, then it is easy to see that transmission of $O(\log (1/\epsilon) d)$ classical bits is sufficient, simply by approximating the real and imaginary parts of the complex entries in the $d$-dimensional state vector. As discussed in Ref.~\cite{Montanaro16q}, this means that the results of Ref.~\cite{Montanaro16q} are asymptotically optimal for bounded error simulation. If exact simulation is required, on the other hand, then Ref.~\cite{Hardy04} shows that without shared randomness, bounded classical communication is insufficient. If exact simulation is required, and shared random data permitted, then surprisingly little is known apart from the exponential lower bounds. Ref.~\cite{Toner03} shows that for $d=2$, transmission of two bits is sufficient (see also Ref.~\cite{Cerf00} which considers positive operator-valued measurements). But we are unaware of any demonstration, even for $d=3$, that classical simulation of arbitrary preparations and projective measurements is possible with bounded communication cost, let alone a demonstration of a specific protocol with bounded communication, or a finite upper bound.

The exponential lower bounds are of interest for the foundations of quantum theory, as well as for communication complexity per se. As Montanaro writes \cite{Montanaro16q}: ``On a fundamental, conceptual level, the question asks: are quantum states `really' like an exponentially-long string of numbers, or do they have a more efficient representation?'' Restated in different language, any lower bound on the size of the classical message in a simulation of the transmission of a quantum state becomes a lower bound on the size of the set of ontic states that the system must have available to it in an ontological model for quantum theory \cite{HarriganSpekkens, PBR12}. For further discussion of the relevance of communication complexity results to quantum foundations, see Ref.~\cite{Montina12}.

The importance of the exponential separation between quantum and classical communication, both for foundations and for complexity theory, underpins one of the motivations of this work, which is to present as simple as possible a proof of this fact, and to obtain as strong a lower bound as possible. In future work, it would be interesting to prove a stronger separation. Using essentially the same proof that we have presented, the following conjecture concerning antidistinguishability would give a separation of $\log d$ qubits vs $\Omega(d \log d)$ bits for exact simulation. This would be asymptotically stronger than all existing results, and would be particularly interesting given the $O(\log (1/\epsilon) d)$ upper bound for bounded-error classical simulation. 
\begin{conjecture}\label{antconj}
Let $|\rho_1\rangle , \ldots , |\rho_d\rangle$ be $d$ pure states. If $|\langle \rho_i | \rho_j\rangle| \leq (d-2)/(d-1)$ for all $i\ne j$, then the states are antidistinguishable.
\end{conjecture}
Some further discussion of and numerical evidence in support of the conjecture is presented in Appendix~C.

\begin{acknowledgments}
\paragraph{Acknowledgement.-}
We thank John-Mark Allen, Matthew Pusey, Matthew Jenssen and Ashley Montanaro for useful discussions. The numerical investigations presented in Appendix~C were carried out by the authors. Some similar investigations were carried out by Matthew Pusey, and we thank him for sharing the results with us. VH is supported by the Clarendon and Keble de Breyne scholarships at the University of Oxford, and an IBM PhD Fellowship. This work was supported by the EPSRC National Quantum Technology Hub in Networked Quantum Information Technologies, and by the Perimeter Institute for Theoretical Physics. Research at Perimeter Institute is supported by the Government of Canada through the Department of Innovation, Science and Economic Development Canada and by the Province of Ontario through the Ministry of Research, Innovation and Science. This project/publication was made possible through the support of a grant from the John Templeton Foundation. The opinions expressed in this publication are those of the author(s) and do not necessarily reflect the views of the John Templeton Foundation. 
\end{acknowledgments}

\bibliographystyle{unsrt}
\bibliography{biblio}

\onecolumngrid

\section{Appendix A: Multiplicative error Sampling} 
The lower bound derived in Claim~\ref{Thm:OneWay} of the main text also straightforwardly applies to multiplicative error sampling.  Let $S$ be the set of states defined in Claim~\ref{Corollary:Triples} of the main text and let $\Pi' = \lbrace \Pi'_z \,| z \in \lbrace j,k,m \rbrace \rbrace$ be a measurement antidistinguishing a triple of states $\lbrace \ket{\rho_j}, \ket{\rho_k}, \ket{\rho_m} \rbrace \subseteq S$, chosen such that $\Pi_{z}' \ket{\rho_z} = 0$ for all $z \in \lbrace j, k, m \rbrace$. Let $p(z \, | \Pi', \ket{\rho_i})$ denote the probability of measuring an outcome `$z$' by applying $\Pi'$  to $\ket{\rho_i}$. A $\epsilon$-multiplicative sampling protocol samples from a distribution $\tilde{p}(z \, | \Pi', \ket{\rho_i})$, such that: \begin{align} \left|\tilde{p}(z \, |   \Pi' , \ket{\rho_i}) - p(z\, | \Pi', \ket{\rho_i}) \right| \leq \epsilon \, p(z\, |\Pi',\ket{\rho_i}), \label{Eq:RelErr} \end{align} for some $\epsilon \geq 0$ and all inputs.
Notice that: 
 \begin{align} p(z|  \, \Pi', \ket{\rho_i})  = 0 &\implies \tilde{p}(z|  \, \Pi', \ket{\rho_i}) = 0, \end{align} 
 holds for arbitrary $\epsilon$. Even as $\epsilon \rightarrow \infty $, the mutliplicative error simulation protocol cannot output an outcome that is assigned zero probability in the exact case and hence also solves the relation $R$ defined in Eq.~\ref{Eq:Relation} of the main text. This means that classical protocols for the sampling task are subject to the same classical lower bounds as the zero-error protocol for $R$.

\section{Appendix B: Alternative simple proof for the existence of exponential-sized complex spherical codes.} 
\label{App:A}

The following argument, from Refs.~\cite{Buhrman01,Alon16}, shows that complex spherical codes exist that are exponentially large in the dimension. Take two random $d$-dimensional pure states:
\begin{equation} 
\begin{aligned} \ket{v} &= \frac{1}{\sqrt{d}}\sum_{i=1}^d v_i \ket{i},& \ket{w} &= \frac{1}{\sqrt{d}}\sum_{i=1}^d w_i\ket{i}, 
\end{aligned} 
\end{equation}
where $w_i, v_i$ are Rademacher random variables with $\text{Pr}(v_i = \pm1) = \text{Pr}(w_i = \pm1) = 1/2$. Their inner product is: 
\begin{align}
\braket{v | w} &= \frac{1}{d} \sum_{i=1}^d v_i w_i = \frac{1}{d} \sum_{i=1}^d X_i,
\end{align}
where $X_i$ is again a Rademacher random variable with $
\text{Pr}(X_i = \pm1) = 1/2.$
The probability that $|\braket{v|w}| > 1/2$ is upper bounded by the Chernoff-Hoeffding inequality: 
\begin{align}
\text{Pr}\left( |\braket{v|w}| > \frac{1}{2} \right) = \text{Pr}\left(\left|\sum_i^d X_i\right| > \frac{d}{2} \right) \leq 2e^{-\frac{d}{8}}. \end{align}
The probability that a set $S$ of such random vectors contains a pair $\ket{v_i}, \ket{w_j} \in S, \; i \neq j$ with overlap greater than $1/2$ is given by:
\begin{align}
\text{Pr}\left( |\braket{v_i | w_j }| > \frac{1}{2} \right) \leq 2 {|S| \choose 2}e^{-\frac{d}{8}}  < |S|^2 e^{-\frac{d}{8}}.
\label{Eq:tailbound}
\end{align}
As soon as this probability falls below $1$, there exists a set $S$ of such random states, so that $|\braket{v_i | w_j}| \leq 1/2$ for any pair. From Eq.~(\ref{Eq:tailbound}), $|S|  = e^{\frac{d}{16}}$. 
Using such a set $S$ for the task defined in the main text, the classical one-way communication complexity is at least $\lceil 0.09\, d -1 \rceil$ bits.

\section{Appendix C:  antidistinguishability conjecture}

Conjecture~\ref{antconj} of the main text, reproduced here, implies a separation of $\log d$ qubits vs $\Omega(d \log d)$ bits for exact simulation. 
\addtocounter{conjecture}{-1}
\begin{conjecture}
Let $|\rho_1\rangle , \ldots , |\rho_d\rangle$ be $d$ pure states. If $|\langle \rho_i | \rho_j\rangle| \leq (d-2)/(d-1)$ for all $i\ne j$, then the states are antidistinguishable.
\end{conjecture}
It is obvious that the conjecture holds with $d=2$, and it follows from the results of Ref.~\cite{Caves02} that the conjecture holds with $d=3$.
In order to gain some intuition for why the conjecture might be true in all dimensions, first consider a generic set of states $|\rho_1\rangle , \ldots , |\rho_d\rangle$. Ref.~\cite{Kerppo18} shows that if any triple of the states is antidistinguishable, then it follows that the set of $d$ states is antidistinguishable. On the other hand, it does not follow that if the set of $d$ states is antidistinguishable, then any triple must be antidistinguishable. Hence the condition that $d$ states is antidistinguishable is logically weaker than the condition that any triple of them is antidistinguishable. If the states satisfy $|\langle \rho_i | \rho_j\rangle| \leq 1/2$ for all $i\ne j$, then the stronger condition holds \cite{Caves02}, hence it is natural to suppose that a similar statement, with $1/2$ on the right hand side replaced by a larger number, suffices for the weaker condition.  

Second, consider the set of $d$-dimensional states:
\begin{eqnarray*}
|\rho_1\rangle &=& \frac{1}{\sqrt{d-1}} \left( |e_2\rangle + |e_3\rangle + \cdots + |e_d\rangle \right) \\
|\rho_2\rangle &=& \frac{1}{\sqrt{d-1}} \left( |e_1\rangle + |e_3\rangle + \cdots + |e_d\rangle \right) \\
&\cdots& \\
|\rho_d\rangle &=& \frac{1}{\sqrt{d-1}} \left( |e_1\rangle + |e_2\rangle + \cdots + |e_{d-1}\rangle \right),
\end{eqnarray*}
where $\{ |e_i\rangle \}$ is an orthonormal basis. This set is antidistinguishable by construction and satisfies $|\langle \rho_i | \rho_j \rangle | = (d-2) / (d-1)$ for all $i\ne j$, which motivates the particular choice of function on the right hand side of the conjectured sufficient condition.

Finally, Fig.~\ref{Fig:Evidence} displays some numerical evidence for Conjecture~\ref{antconj}. Note that for a given set of states in $d$ dimensions, determining whether they are antidistinguishable or not corresponds to solving a semi-definite program (SDP) \cite{Bandyopadhyay14}. For each $d=2,3,4,5$, we generated a set of
  $150000$ sets of $d$ vectors in $d$ dimensions, where each vector is chosen independently, and uniformly according to the Haar measure. For each set of vectors, the SDP is solved, in order to determine whether the set is antidistinguishable. For those sets that are not antidistinguishable, the quantity $\alpha = \max_{i\ne j} |\langle \rho_i | \rho_j \rangle |$ is recorded. The shaded region of the graph shows, for each dimension, the minimum value of $\alpha$ obtained, and the dashed line shows the value of $(d-2)/(d-1)$, with lines rather than points used for clarity. A counterexample to the conjecture would appear as the dashed line crossing the non-shaded region. The evidence for the conjecture consists in the fact that the shaded region cleaves fairly closely to the dashed line, yet no counterexamples were found.
\begin{figure*}[h]
\includegraphics[scale=.6]{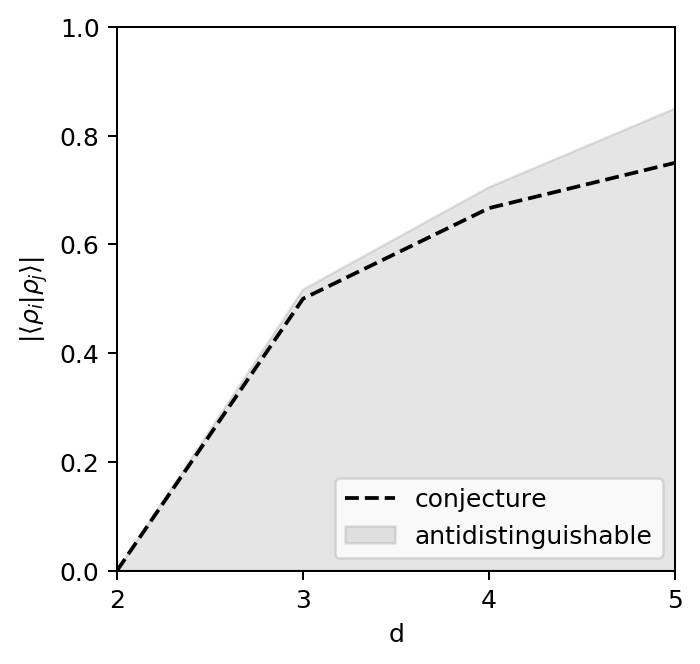}
\caption{Numerical evidence for Conjecture~\ref{antconj} of the main text. }
\label{Fig:Evidence}
\end{figure*}

\end{document}